\documentclass{article}
\usepackage{graphicx}
\usepackage{natbib}
\usepackage{multirow}
\usepackage{amsmath}
\usepackage{amsthm}
\usepackage{amssymb}
\usepackage{float}
\usepackage{authblk}
\usepackage{comment}
\usepackage{enumitem}
\usepackage{hyperref}

\newtheorem{theorem}{Theorem}
\newtheorem{lemma}[theorem]{Lemma}

\usepackage[font=small]{caption}

\title{A Novel Multiscale Framework for Testing Independence: Efficient Detection of Explicit or Implicit Functional Relationships}
\author[1]{Seetharaman P}
\author[2]{Sagnik Das}
\author[1]{Angshuman Roy}
\affil[1]{Department of Mathematics and Statistics, Indian Institute of Technology, Tirupati}
\affil[2]{Indian Institute of Science Education and Research, Kolkata}
\date{}

\begin{document}

\maketitle
\begin{abstract}
   In this article, we consider the problem of testing the independence between two random variables.
   Our primary objective is to develop tests that are highly effective at detecting associations arising from explicit or implicit functional relationship between two variables.
   We adopt a multiscale approach by analyzing neighborhoods of varying sizes within the dataset and aggregating the results.
   We introduce a general testing framework designed to enhance the power of existing independence tests to achieve our objective.
   Additionally, we propose a novel test method that is powerful as well as computationally efficient.
   The performance of these tests is compared with existing methods using various simulated datasets.
   Additionally, a visualization method has been proposed for exploring the localization of dependence within datasets.
\end{abstract}

\section{Introduction}
Tests of independence play a vital role in statistical analysis.
They are used to determine relationships between variables, validate models, select relevant features from a large pool of features, and establish causal directions, among other applications.
These tests are particularly important in fields such as economics, biology, social sciences, and clinical trials.

The mathematical formulation for testing independence is as follows. Consider two random variables \(X\) and \(Y\) with distribution functions \(F_X\) and \(F_Y\), respectively, and let \(F\) represent their joint distribution function. The objective is to test the null hypothesis \(H_0\) against the alternative hypothesis \(H_1\) based on \(n\) independent and identically distributed (i.i.d.) observations \((x_1, y_1), (x_2, y_2), \ldots, (x_n, y_n)\) drawn from \(F\), where
\begin{equation}
    \begin{cases} 
    H_0: F = F_XF_Y \\
    H_1: F \neq F_XF_Y.
    \end{cases}
    \label{eq:test}
\end{equation}

A substantial amount of research has been conducted on this topic, and it remains an active area of investigation.
Pearson's correlation is perhaps the most well-known classical measure that quantifies linear dependence.
Spearman's rank correlation \citep{spearman1904proof} and Kendall's concordance-discordance statistic \citep{kendall1938new} are used to measure monotonic associations.
\cite{hoeffding1948non} proposed a measure using the \(L_2\) distance between the joint distribution function and the product of marginals.
\cite{szekely2007measuring} introduced distance correlation (dCor), based on energy distance, while \cite{gretton2007kernel} leveraged kernel methods to develop the Hilbert-Schmidt independence criterion (HSIC).
\cite{reshef2011detecting} proposed the maximal information coefficient (MIC), an information-theoretic measure that reaches its maximum when one variable is a function of the other.
\citep{heller2013consistent} introduced a powerful test that breaks down the problem into multiple \(2\times 2\) contingency table independence tests and aggregates the results.
\cite{zhang2019bet} proposed a test that is uniformly consistent with respect to total variation distance.
\cite{roy2020somea}, and \cite{roy2020someb} used copula and checkerboard copula approaches to propose monotonic transformation-invariant tests of dependence for two or more variables.
\cite{chatterjee2021new} introduced an asymmetric measure of dependence with a simple form and an asymptotic normal distribution under independence, which attains its highest value only when a functional relationship exists.
These are just a few of the many contributions in this field.

Different tests of independence have unique strengths and limitations.
As demonstrated in Theorem 2.2 of \cite{zhang2019bet}, no test of independence can achieve uniform consistency.
This implies that for any fixed sample size, the power of a test cannot always exceed the nominal level across all alternatives.
Therefore, selecting the most appropriate test for a given scenario is crucial.
For instance, the correlation coefficient is particularly powerful for detecting dependence between two jointly normally distributed random variables.
It is highly effective at identifying linear relationships.
Similarly, Spearman's rank correlation is particularly suited for detecting monotonic relationships.

In this article, we aim to develop tests of independence (as described in Equation \ref{eq:test}) that will be particularly powerful for detecting association between two random variables, \(X\) and \(Y\), that adhere to the following parametric equation:
\begin{equation}
    \begin{cases} 
    X = f(Z) + \epsilon_X \\
    Y = g(Z) + \epsilon_Y,
    \end{cases}
    \label{eq:param}
\end{equation} where \(f\) and \(g\) are continuous functions not constant on any interval, \(Z\) is a continuous random variable defined on an interval, and \(\epsilon_X\) and \(\epsilon_Y\) are independent noise components that are each independent of \(Z\).
As a result, our tests will also be powerful for detecting association between random variables \(X\) and \(Y\) that are functionally related to each other, i.e., when \(Y = f(X) + \epsilon_Y\) or \(X = g(Y) + \epsilon_X\).
In contrast to the method proposed by \cite{chatterjee2021new}, which is powerful for detecting dependence when \(Y\) is a function of \(X\), in our approach, \(X\) and \(Y\) are treated symmetrically.

The article is organized as follows.
In Section 2, we introduce a general testing framework aimed at enhancing the power of existing tests of independence in scenarios described by Equation \eqref{eq:param} and examine its computational complexity. 
In Section 3, we apply this framework to develop a novel test of independence and analyze its computational complexity.
Section 4 presents a performance comparison of our methods on various simulated datasets.
Section 5 extends this analysis to a real-world dataset.
In Section 6, we demonstrate a visualization technique for exploring the localization of dependence in a real-world dataset.
Finally, we conclude with a discussion and summary of our findings.

\section{A General Multiscale Testing Framework}

Let us begin this section with a few examples of continuous dependent random variables \(X\) and \(Y\) that satisfy the parametric equation described in Equation \ref{eq:param}.
We consider three examples, referred to as Example A, B, and C.
\begin{itemize}
    \item Example A: \(X = \Theta + \epsilon_X\) and \(Y = \sin(\Theta) + \epsilon_Y\), where \(\Theta\), \(\epsilon_X\), and \(\epsilon_Y\) are mutually independent, with \(\Theta \sim U(0,2\pi)\) and \(\epsilon_X, \epsilon_Y \sim N(0, (1/20)^2)\).
    \item Example B: \(X = \cos(\Theta) + \epsilon_X\) and \(Y = \sin(\Theta) + \epsilon_Y\).
    \item Example C: \(Y = U + \epsilon_Y\) and \(X = U^2 + \epsilon_X\), where \(U \sim U(-1,1)\) and \(U\) is independent of \(\epsilon_X\) and \(\epsilon_Y\).
\end{itemize}
The scatter plots for these examples are shown in Figure \ref{fig:exampleABC}.
A common feature of all these examples is that within certain neighborhoods around support points, the conditional correlation between \(X\) and \(Y\) is non-zero.
Specifically, there exist points \((x_0, y_0)\) and \((x', y')\) in the support of the distribution such that \(\text{Cor}(X,Y \mid (X,Y) \in N(x,y; |x' - x_0|, |y' - y_0|)\) is non-zero, where \(N(x,y;\epsilon_1,\epsilon_2)\) represents the \(\epsilon_1,\epsilon_2\)-neighborhood of \((x,y)\), defined as \([x-\epsilon_1,x+\epsilon_1] \times [y-\epsilon_2, y+\epsilon_2]\).
In Figure \ref{fig:exampleABC}, red rectangles highlight instances of such neighborhoods with non-zero conditional correlation.
\begin{figure}[h]
\centering
\includegraphics[width=0.75\linewidth]{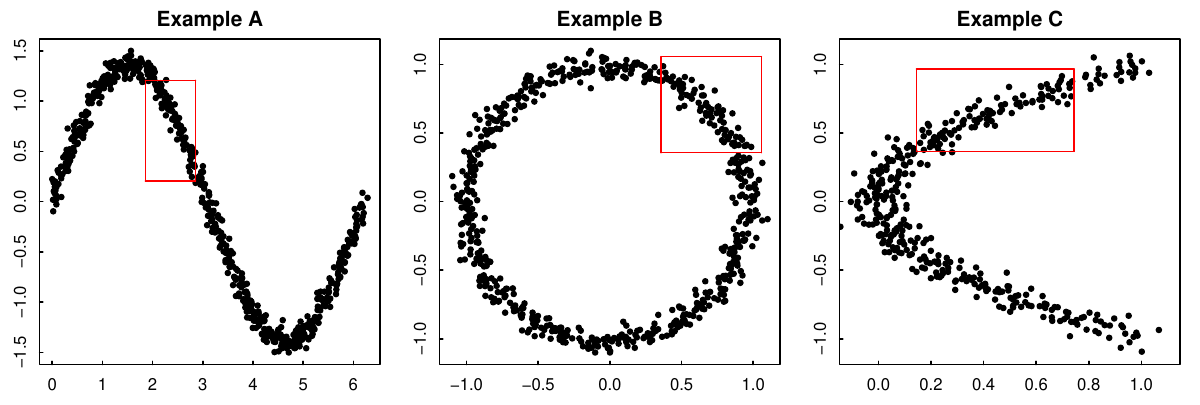}
\caption{Neighborhoods having non-zero conditional correlation in Example A, B, and C.}
\label{fig:exampleABC}
\end{figure}

As the next example, we consider the \(BEX_d\) distribution, originally introduced in \citep{zhang2019bet} where it is defined as the uniform distribution over a set of parallel and intersecting lines given by \(\sum_{i=1}^{2^d-1}\sum_{j=1}^{2^d-1}(|x-c_i|-|y-c_j|)I[|x-c_i|\leq 2^{-d}, |y-c_j|\leq 2^{-d}]=0\), where \(c_i=(2i-1)/2^d\).
Figure \ref{fig:bex_plot} illustrates the \(BEX_d\) distribution for \(d=1,2\), and \(3\).
If \((X,Y)\) follows the \(BEX_d\) distribution, the marginals \(X\) and \(Y\) are dependent but both follow a continuous uniform distribution over \([0,1]\).
Moreover, \(X\) and \(Y\) can be shown to satisfy Equation \ref{eq:param}.
Although \(\text{Cor}(X,Y)=0\) for the \(BEX_d\) distribution, there exist neighborhoods around support points where the conditional correlation is extremely high.
\begin{figure}[h]
\centering
\includegraphics[width=0.75\linewidth]{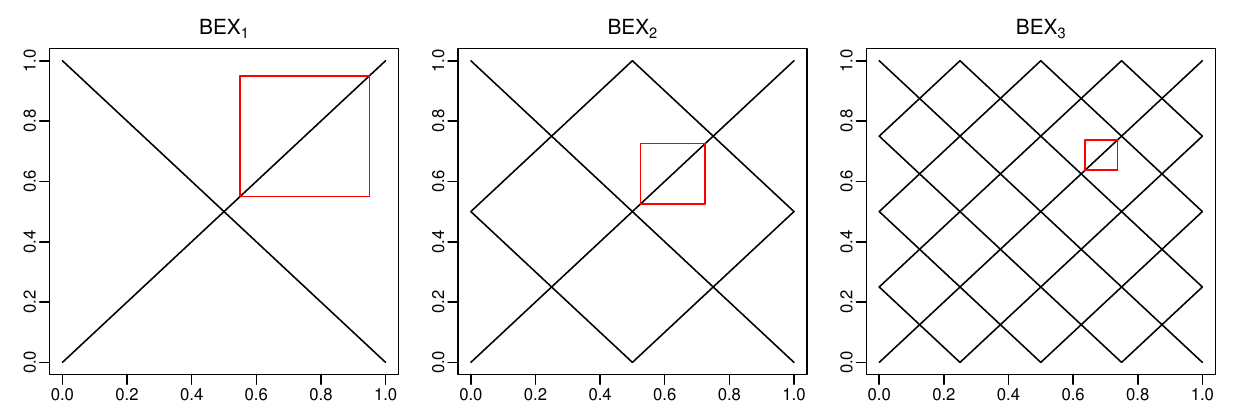}
\caption{Existence of neighborhoods having high conditional correlation values in \(BEX_d\) distributions.}
\label{fig:bex_plot}
\end{figure}
Red rectangles highlight some such neighborhoods with high conditional correlation values in \(BEX_1\), \(BEX_2\), and \(BEX_3\) distributions in Figure \ref{fig:bex_plot}.

These examples suggest that by calculating the test statistics for an existing test of independence across different neighborhoods and aggregating these values meaningfully, we can enhance its power, especially in scenarios described by Equation \eqref{eq:param}.
Building on this idea, we introduce a multiscale approach for existing tests of independence in this section.

Let \(\mathbf{xy}_{1:n}\) denote the \(n\) i.i.d. observations \((x_1,y_1),\ldots,(x_n,y_n)\).
We analyze the dataset using neighborhoods of varying sizes centered at each observation point, with the distances from the center point to other observations serving as our guide for selecting neighborhood sizes.
For a given sample \(\mathbf{xy}_{1:n}\) from a bivariate distribution \(F\), we consider all neighborhoods of the form  \(N(x_i,y_i;|x_j-x_i|,|y_j-y_i|)\) for \(i\neq j\).
Thus, we consider a total of \(n(n-1)\) neighborhoods.
For notational convenience, we denote \(N(x_i,y_i;|x_j-x_i|,|y_j-y_i|)\) by \(N_{i,j}\).

Let \(T\) be a test statistic for testing independence between two univariate random variables.
Let us use the notation \(T(\mathbf{xy}_{1:n})\) to denote the value of the test statistic \(T\) computed on the sample \(\mathbf{xy}_{1:n}\).
Here, we consider only those test statistics \(T\) such that \(T(\mathbf{xy}_{1:n})\geq 0\), \(T(\mathbf{xy}_{1:n})\rightarrow 0\) in probability as \(n\rightarrow \infty\) under independence, and the rejection region is of the form \(\{T(\mathbf{xy}_{1:n})>C_n\}\) for some constant \(C_n>0\).
For example, Person's correlation coefficient cannot be considered as \(T\), but its absolute value can be used.
We define \(T_{i,j}:=T(\mathbf{xy}_{1:n}\cap N_{i,j})\), that is, \(T_{i,j}\) is the value of the test statistic \(T\) when evaluated on the observations that fall within the neighborhood \(N_{i,j}\).
If \(T\) is undefined on \(\mathbf{xy}_{1:n}\cap N_{i,j}\) (it may happen when \(x_i=x_j\) or \(y_i=y_j\)), we set \(T_{i,j}:=0\).
One naive way to aggregate all the findings from different neighborhoods is to sum \(T_{i,j}\)'s, that is, to consider \(\sum_{1\leq i\neq j\leq n} T_{i,j}\) as our test statistics.
The problem with this summation is that if the dependency information is limited to relatively small neighborhoods, summing all the \(T_{i,j}\) values might introduce noise into this information.
To address this issue, we separate \(T_{i,j}\)'s in \((n-1)\) distinct groups according to the proximity between \((x_i,y_i)\) and \((x_j,y_j)\).
A detailed description is as follows.

For each observation \((x_i,y_i)\), we order the remaining observations according to their Euclidean distance from \((x_i,y_i)\).
Random tie-breaking is to be used in case of ties while ordering.
Let \((x_{\pi_i(1)},y_{\pi_i(1)}),\) \(\ldots,\) \((x_{\pi_i(n-1)},y_{\pi_i(n-1)})\) be the observations ordered by their Euclidean distance from \((x_i,y_i)\) in ascending order.
Thus, \(N_{i,\pi_i(k)}\) is the neighborhood of \((x_i,y_i)\) that has \(k\)-th nearest neighbor of \((x_i,y_i)\) at one of its vertices.
It is easy to see that for a fixed \(i\), the length of the diagonal of \(N_{i,\pi_i(k)}\) increases as \(k\) increases.
We average the value of the test statistic \(T\) based on the sample \(\mathbf{xy}_{1:n}\cap N_{i,\pi_i(k)}\) keeping \(k\) fixed and varying over \(i\) over \(1,\ldots,n\).
We denote this average as \(T_{[k]}\), that is, \(T_{[k]}:=n^{-1}\sum_{i=1}^nT(\mathbf{xy}_{1:n}\cap N_{i,\pi_i(k)})=n^{-1}\sum_{i=1}^n T_{i,\pi_i(k)}\).
To determine how extreme the value of \(T_{[k]}\) is, we need to know its distribution under the independence of the marginals, which can be estimated by resampling technique.
We discuss this in detail next.

\begin{figure}[h]
\centering
\includegraphics[width=0.4\linewidth]{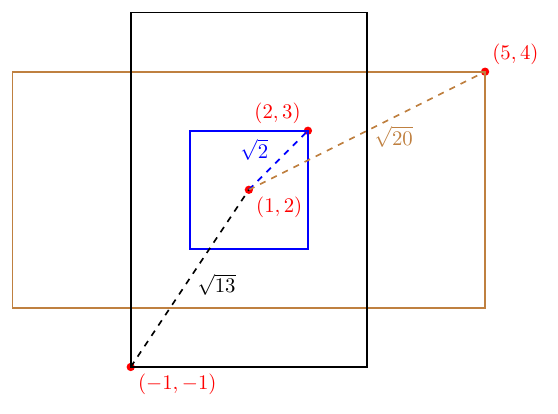}
\caption{Given the four observation points \((1,2)\), \((2,3)\), \((-1,-1)\), and \((5,4)\), the areas enclosed by the blue, black, and orange rectangles represent the neighborhoods selected around the observation point \((1,2)\).}
\label{fig:neighborhood_plot}
\end{figure}

Given a sample \(\mathbf{xy}_{1:n}\), a randomly permuted sample \(\mathbf{xy}^{(\tau)}_{1:n}\) can be generated, where \(\tau(1),\ldots,\tau(n)\) is a random permutation of \(1,\ldots,n\) and \(\mathbf{xy}^{(\tau)}_{1:n}:=\{(x_1,y_{\tau(1)}),\ldots,(x_n,y_{\tau(n)})\}\).
We calculate \(T_{[1]},\ldots,T_{[n-1]}\) in the same way on the sample \(\mathbf{xy}^{(\tau)}_{1:n}\) and denote their values with \(T_{[1]}^{(\tau)},\ldots,T_{[n-1]}^{(\tau)}\).
For \(B\) independent random permutations \(\tau_1,\ldots,\tau_B\), we can compute \(T_{[k]}^{(\tau_1)},\ldots,T_{[k]}^{(\tau_B)}\) for \(1\leq k\leq n-1\).
According to the permutation test principle, the empirical distribution of \(T_{[k]}^{(\tau_1)},\ldots,T_{[k]}^{(\tau_B)}\) is an estimator of the distribution of \(T_{[k]}\) under independence.
We can estimate the mean of \(T_{[k]}\) under independence by \(\bar{T}_{[k]}^{H_0}:=B^{-1}\sum_{i=1}^B T_{[k]}^{(\tau_i)}\) and the standard deviation of \(T_{[k]}\) under independence by \(s_{[k]}^{H_0}:=\sqrt{B^{-1}\sum_{i=1}^B (T_{[k]}^{(\tau_i)}-\bar{T}_{[k]}^{H_0})^2}\).
Next we compute Z-score of \(T_{[k]}\) with respect to its distribution under independence by \(z_k:=\frac{T_{[k]}-\bar{T}_{[k]}^{H_0}}{s_{[k]}^{H_0}}\).
If \(s_{[k]}^{H_0}=0\), we define \(z_k=0\).
The Z-scores \(z_1,\ldots,z_{n-1}\) suggest how extreme the values of \(T_{[1]},\ldots,T_{[n-1]}\) are.
Analyzing them for a sample \(\mathbf{xy}_{1:n}\) can give us valuable insight into the dependence structure.
We demonstrate this point below using various bivariate distributions.

In this illustration, we consider two popular test statistics as \(T\): absolute value of Pearson's correlation and distance correlation.
We denote these two test statistics by \(T^{cor}\) and \(T^{dcor}\) respectively.
We consider eight bivariate distributions, and a description of each is provided in Table \ref{tab:zk_plots_datasets}.
\begin{table}[h]
\centering
\scalebox{0.85}{
\begin{tabular}{l l}
\hline
\hline
Distribution & Description \\
\hline
\hline
(a) Square & \(X,Y\stackrel{\text{i.i.d.}}{\sim} U(-1,1)\)\\
\hline
(b) Straight line & \(X\sim U(-1,1), Y=X\)\\
\hline
(c) Noisy straight line & \(U\sim U(-1,1), X=U+e_X, Y=U+e_Y,\)\\
& \(e_X,e_Y\stackrel{\text{i.i.d.}}{\sim} N(0,0.1)\)\\
\hline
(d) Sine & \(X\sim U(0,2\pi), Y=\sin(5X),\)\\
\hline
(e) Circle & \(\Theta\sim U(0,2\pi), X=\cos(\Theta), Y=\sin(\Theta)\) \\
\hline
(f) Noisy parabola & \(U\sim U(-1,1), X=U^2+\epsilon_X, Y=U+\epsilon_Y,\) \\
& \(e_X,e_Y\stackrel{\text{i.i.d.}}{\sim} N(0,0.25)\)\\
\hline
(g) \(BEX_2\) & As described in Section 2\\
\hline
(h) BVN, rho=0.5 & \((X,Y)\) follows standard bivariate normal\\
& distribution with a correlation coefficient of 0.5\\
\hline
\hline
\end{tabular}
}
\caption{Description of distributions considered}
\label{tab:zk_plots_datasets}
\end{table}
The scatter plots of these distributions are presented in Figure \ref{fig:zk_plots_datasets}.
It is easy to see that except for the `Square' distribution, \(X\) and \(Y\) are dependent in all the other distributions.

From each distribution, we generated 50 i.i.d. observations and calculated \(z_1,\ldots,z_{49}\).
By repeating this step independently 1000 times, we then calculated the average values \(\bar{z}_1,\ldots,\bar{z}_{49}\).
We plotted these average values for each distribution in Figure \ref{fig:zk_plots}.
Under independence, it is evident that the expected value of \(z_k\) is 0.
Therefore, if \(z_k\) deviates from 0, it will be indicative of dependence.
\begin{figure}[h]
\centering
\includegraphics[width=0.85\linewidth]{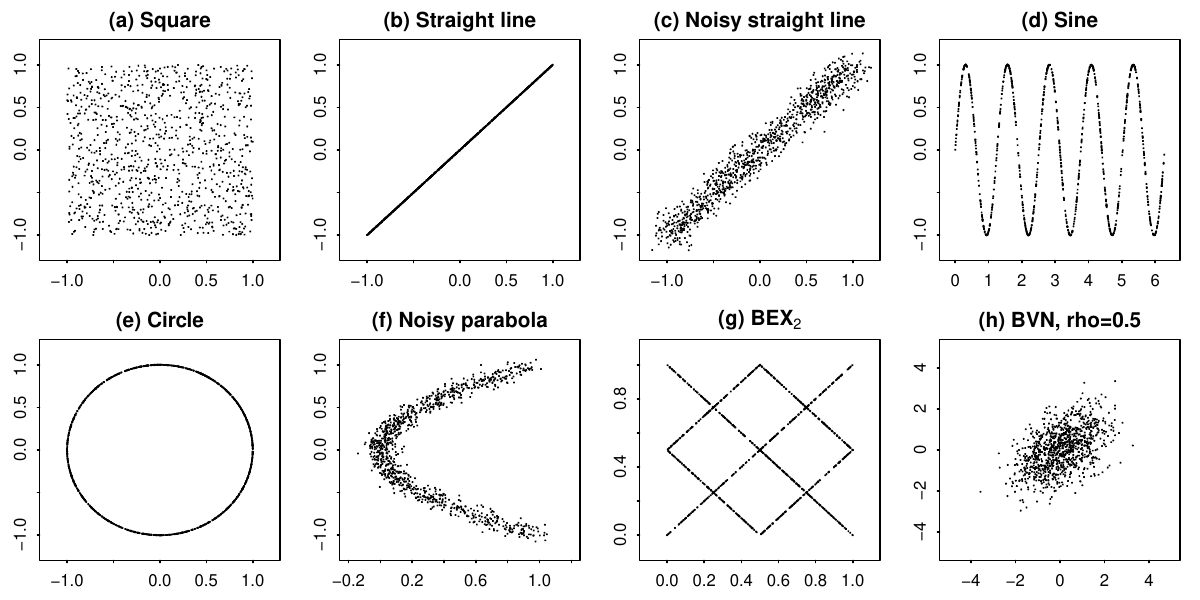}
\caption{Scatter plots of different datasets used}
\label{fig:zk_plots_datasets}
\end{figure}

\begin{figure}[h]
\centering
\includegraphics[width=0.85\linewidth]{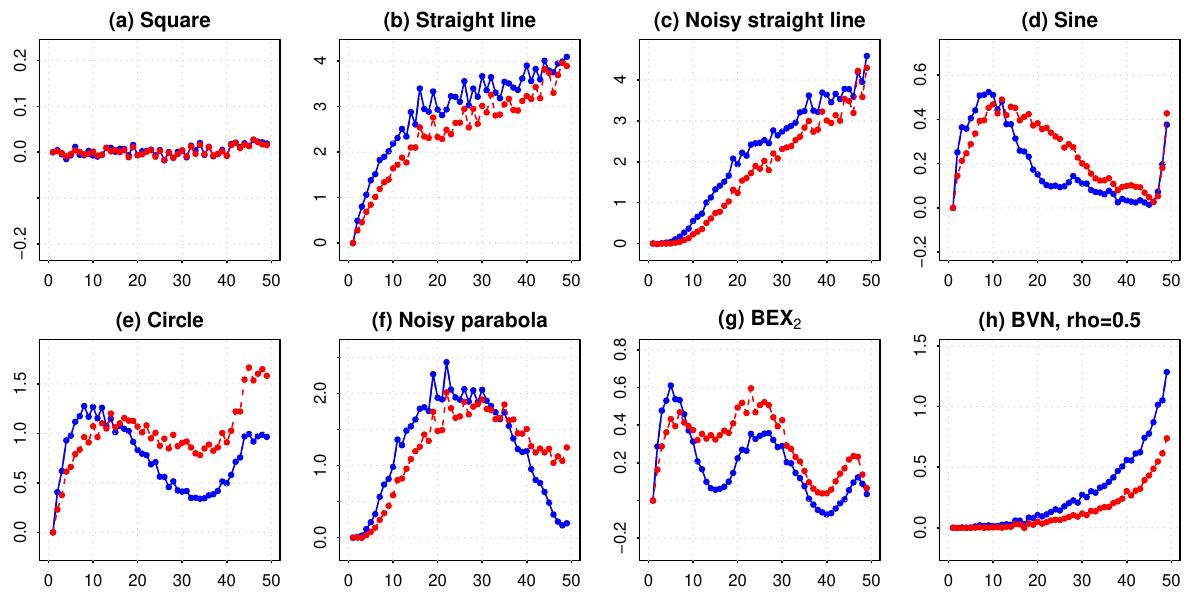}
\caption{The plot of \(\bar{z}_1,\ldots,\bar{z}_{49}\) for different distributions. The blue and red points are corresponding to \(T^{cor}\) and \(T^{dcor}\) statistics respectively.}
\label{fig:zk_plots}
\end{figure}
From Figure \ref{fig:zk_plots}, we observe that the average Z-score values are close to 0 for `Square' distribution as \(X\) and \(Y\) are independent.
For the `Straight line' and `Noisy straight line' distributions, the average Z-scores are higher in larger neighborhoods for both \(T^{cor}\) and \(T^{dcor}\).
However, in the `Noisy straight line', the dependence information is less apparent in smaller neighborhoods compared to the `Straight line'.
For the `Sine' distribution, the dependence information is clearly noticeable in smaller neighborhoods. 
An interesting phenomenon occurs with the `Circle' distribution, where most neighborhoods, including the largest ones, contain arcs, allowing both \(T^{cor}\) and \(T^{dcor}\) to detect dependence effectively in most neighborhoods.
In the `Noisy parabola', the dependence is most prominent in mid-sized neighborhoods.
For the more complex `\(BEX_2\)' distribution, these statistics detect the dependence patterns in smaller and mid-sized neighborhoods, but less so in larger ones.
Finally, in `BVN, rho=0.5', the dependence information is only clearly evident in larger neighborhoods.

It follows clearly from the above illustration that test statistics such as absolute value of correlation and distance correlation sometimes detect dependence information better in smaller neighborhoods, sometimes in mid-sized, and sometimes in larger neighborhoods.
Therefore, it makes sense to aggregate information from all the Z-scores \(z_1,\ldots,z_{n-1}\).
We propose an aggregation method in the following paragraph.

First, we observe that under independence, the expected value of \(z_k\) is \(0\) for \(k=1,\ldots,n-1\).
Under dependence, if the \(T_{i,j}\)'s are stochastically larger, \(T_{[k]}\) will also be stochastically larger.
In that case, the expected values of \(z_k\) will also be positive.
Let \(\mu_k\) be the expected value of \(z_k\), that is, \(\mu_k=E[z_k]\).
Thus testing \(H_0\) against \(H_1\) can be done by testing \(H_0^{'}: \mu_k=0~\forall k\) vs. \(H_1^{'}: \exists~ k~\text{such that}~\mu_k>0\).
To test this, we suggest the following test statistic \(\Psi_n:=\sum_{k=1}^{n-1}(\max\{z_k,0\})^2\).
Clearly, a high value of \(\Psi_n\) presents evidence against the null hypothesis.
To determine the distribution of \(\Psi_n\) under \(H_0\), we again use resampling approach by utilizing \(B\) randomly permuted samples \(\mathbf{xy}_{1:n}^{(\tau_1)},\ldots,\mathbf{xy}_{1:n}^{(\tau_B)}\) which are already in our disposal.
Similar how we computed \(z_1,\ldots,z_{n-1}\) for \(\mathbf{xy}_{1:n}\) using permuted samples \(\mathbf{xy}_{1:n}^{(\tau_1)},\ldots,\mathbf{xy}_{1:n}^{(\tau_B)}\), we compute \(z_1^{(\tau_i)},\ldots,z_{n-1}^{(\tau_i)}\) for \(\mathbf{xy}_{1:n}^{(\tau_i)}\) based on permuted samples \(\mathbf{xy}_{1:n}^{(\tau_1)},\ldots,\mathbf{xy}_{1:n}^{(\tau_{i-1})},\mathbf{xy}_{1:n}^{(\tau_{i+1})},\mathbf{xy}_{1:n}^{(\tau_B)}\) for \(i=1,\ldots,B\).
Next we calculate \(\Psi_n^{(\tau_i)}:=\sum_{k=1}^{n-1}(\max\{z_k^{(\tau_i)},0\})^2\) for \(i=1,\ldots,B\).
Finally we determine the p-value of \(\Psi_n\) by \(p=B^{-1}\sum_{i=1}^B I[\Psi_n \leq \Psi_n^{(\tau_i)}]\).
We reject \(H_0\) if the p-value turns out to be less than the level of significance.

Assume that computing \(T(\mathbf{xy}_{1:n})\) requires \(\mathcal{O}(\Theta(n))\) operations.
In that case, computing \(T_{i,j}\) collectively for \(1\leq i\neq j\leq n\) requires \(\mathcal{O}( n^2 \Theta(n))\) operations.
For a fixed \(1 \leq i \leq n\), the merge sort algorithm allows us to compute \(\pi_i\) in \(\mathcal{O}(n \log n)\) operations.
Therefore, determining \(T_{[1]},\ldots,T_{[n-1]}\) collectively takes \(\mathcal{O}(n^2\Theta(n) + n^2 \log n)\) operations (recall that  \(T_{[k]} = n^{-1} \sum_{i=1}^n T_{i,\pi_i(k)}\)).
As a result, the computational complexity of computing \(z_1,\ldots,z_{n-1}\) together also is \(\mathcal{O}(n^2\Theta(n) + n^2 \log n)\) assuming that number of randomly permuted samples \(B\) is constant with respect to \(n\).
Thus we conclude that \(\mathcal{O}(n^2\Theta(n) + n^2 \log n)\) operations are required for computing the test statistics \(\Psi_n = \sum_{k=1}^{n-1} (\max\{z_k, 0\})^2\).

Since the correlation coefficient can be calculated in linear time, \(\Theta(n) = n\) when \(T = T^{cor}\).
Thus, the complexity of calculating \(\Psi_n\) when \(T = T^{cor}\) is \(\mathcal{O}(n^3)\).
On the other hand, the distance correlation between two univariate random variables can be computed in \(\mathcal{O}(n \log n)\) operations \citep[see,][]{chaudhuri2019fast}.
Therefore, the complexity of calculating \(\Psi_n\) when \(T = T^{dcor}\) is \(\mathcal{O}(n^3 \log n)\).\\

\noindent\fbox{%
\parbox{\textwidth}{%
\underline{Algorithm for General Multiscale Testing Framework}\\

Input: set of bivariate observations \(xy_{1:n}\), test statistics \(T\), number of permutations \(B\)\\

\begin{enumerate}[topsep=0pt,itemsep=-1ex,partopsep=1ex,parsep=1ex]
    \item Set \(i=1\).
    \item For \(k = 1, \dots, n-1\), determine \(\pi_i(k)\) such that \((x_{\pi_i(k)},y_{\pi_i(k)})\in xy_{1:n}\) is the \(k\)-th nearest neighbor of \((x_i,y_i)\) in terms of Euclidean distance. Here \((x_i,y_i)\) is the \(i\)-th observation in \(xy_{1:n}\).
    \item If \(i<n\), increase \(i\) by 1 and go to Step 2.
    \item For \(k=1,\ldots,n-1\), compute \(
    T_{[k]} := \frac{1}{n} \sum_{i=1}^n T(xy_{1:n} \cap N_{i, \pi_i(k)}),\) where \(N_{i, \pi_i (k)}:=[x_i-|x_i-x_{\pi_i(k)}|,x_i+|x_i-x_{\pi_i(k)}|]\times[y_i-|y_i-y_{\pi_i(k)}|,y_i+|y_i-y_{\pi_i(k)}|]\).
    \item Set \(b=1\).
    \item Generate  \(xy^{(\tau_b)}_{1:n} := \{(x_1, y_{\tau_b(1)}), \dots, (x_n, y_{\tau_b(n)})\}\), where \(\tau_b\) is a random permutation on \(\{1, \dots, n\}\). Here \((x_i,y_i)\) is the \(i\)-th observation in \(xy_{1:n}\).
    \item Perform Steps 1 to 4 using \(xy^{(\tau_b)}_{1:n}\) instead of \(xy_{1:n}\) to compute \(T^{(\tau_b)}_{[1]},\ldots,T^{(\tau_b)}_{[n-1]}\).
    \item If \(b<B\), increase \(b\) by 1 and go to Step 6.
    \item Calculate \(\bar{T}^{H_0}_{[k]} := \frac{1}{B} \sum_{i=1}^B T^{(\tau_i)}_{[k]}\) and \(s^{H_0}_{[k]} := \sqrt{\frac{1}{B-1} \sum_{i=1}^B \left(T^{(\tau_i)}_{[k]} - \bar{T}^{H_0}_{[k]}\right)^2}\).
    \item For \(k = 1, 2, \dots, n-1\), compute \(z_k := \frac{T_{[k]} - \bar{T}^{H_0}_{[k]}}{s^{H_0}_{[k]}}\). Set \(z_k: = 0\) if \(s^{H_0}_{[k]} = 0\).
    \item Compute \(\Psi_n := \sum_{k=1}^{n-1} \left(\max\{z_k, 0\}\right)^2\).
    \item Reset \(b=1\).
    \item For \(k = 1, 2, \dots, n-1\), compute \(z_k^{(\tau_b)} := \frac{T^{(\tau_b)}_{[k]} - \bar{T}^{H_0}_{[k]}}{s^{H_0}_{[k]}}\). Set \(z_k^{(\tau_b)}: = 0\) if \(s^{H_0}_{[k]} = 0\).
    \item Compute \(\Psi_n^{(\tau_b)} := \sum_{k=1}^{n-1} \left(\max\{z_k^{(\tau_b)}, 0\}\right)^2\).
    \item If \(b<B\), increase \(b\) by 1 and go to Step 13.
    \item Determine the p-value of \(\Psi_n\) as: \(p = \frac{1}{B} \sum_{i=1}^B I[\Psi_n \leq \Psi^{(\tau_i)}_n]\).
\end{enumerate}
}%
}

\section{A Special Test Method}

In this section, we introduce a testing method that closely follows the framework proposed in the previous section, with a few modifications.
\(X\) and \(Y\) are assumed to be continuous random variables in this setup.
The key distinction in this approach lies in the manner in which the \(T_{i,j}\) values are computed.

Given a sample \(\mathbf{xy}_{1:n}\), we begin by fixing two observations \((x_i, y_i)\) and \((x_j, y_j)\), where \(i \neq j\).
Similar to the previous method, we define the neighborhood \(N_{i,j} = N(x_i, y_i; |x_j - x_i|, |y_j - y_i|)\), centered at \((x_i, y_i)\) with \((x_j, y_j)\) positioned at a corner of the rectangle defined by the neighborhood.
Within this neighborhood, we partition the space into four quadrants, classifying observations according to whether the \(x\)-coordinate is less than or greater than \(x_i\) and whether the \(y\)-coordinate is less than or greater than \(y_i\).
Subsequently, a \(2 \times 2\) contingency table is constructed to count the number of observations in each of these quadrants (see Table \ref{tab:cont_table}).
\begin{table}[h]
\centering
\scalebox{0.85}{
\begin{tabular}{c|c|c}
\hline 
  &  \(x_k < x_i\)  & \(x_k > x_i\) \\
\hline 
\(y_k > y_i\)  & \(a_{i,j}\) & \(b_{i,j}\) \\
\(y_k < y_i\)  & \(c_{i,j}\) & \(d_{i,j}\) \\
\hline 
\end{tabular}
}
\caption{Number of \((x_k, y_k)\) that belong to one of the four quadrants of \(N_{i,j}\).}
\label{tab:cont_table}
\end{table}

Let \(a_{i,j}, b_{i,j}, c_{i,j}, d_{i,j}\) represent the frequencies in this contingency table, where \(a_{i,j} = \#\{(x_k, y_k) : (x_k, y_k) \in \mathbf{xy}_{1:n} \cap N_{i,j}, x_k < x_i, y_k > y_i\}\), \(b_{i,j} = \#\{(x_k, y_k) : (x_k, y_k) \in \mathbf{xy}_{1:n} \cap N_{i,j}, x_k > x_i, y_k > y_i\}\), \(c_{i,j} = \#\{(x_k, y_k) : (x_k, y_k) \in \mathbf{xy}_{1:n} \cap N_{i,j}, x_k < x_i, y_k < y_i\}\), and \(d_{i,j} = \#\{(x_k, y_k) : (x_k, y_k) \in \mathbf{xy}_{1:n} \cap N_{i,j}, x_k > x_i, y_k < y_i\}\).

\(T_{i,j}\) is defined as the absolute value of the phi coefficient calculated from this contingency table:
\[
T_{i,j} := \frac{|a_{i,j} d_{i,j} - b_{i,j} c_{i,j}|}{\sqrt{(a_{i,j} + b_{i,j})(c_{i,j} + d_{i,j})(a_{i,j} + c_{i,j})(b_{i,j} + d_{i,j})}}.
\]
If \((a_{i,j} + b_{i,j})(c_{i,j} + d_{i,j})(a_{i,j} + c_{i,j})(b_{i,j} + d_{i,j}) = 0\), we set \(T_{i,j} = 0\).

Next, we employ the same framework described in the previous section to determine a p-value.
For notational convenience, the test statistic \(T\) used here will be referred to as \(T^{phi}\).

Following the setup in Figure \ref{fig:zk_plots}, we plotted the average Z-scores for the \(T^{phi}\) statistic across various neighborhood sizes, as shown in Figure \ref{fig:zk_plots2} to gain a deeper understanding of its underlying mechanism.
These values are represented by red dots.
For comparison, the average Z-scores for the \(T^{cor}\) statistic are plotted alongside and represented by the blue dots.

\begin{figure}[h]
\centering
\includegraphics[width=0.85\linewidth]{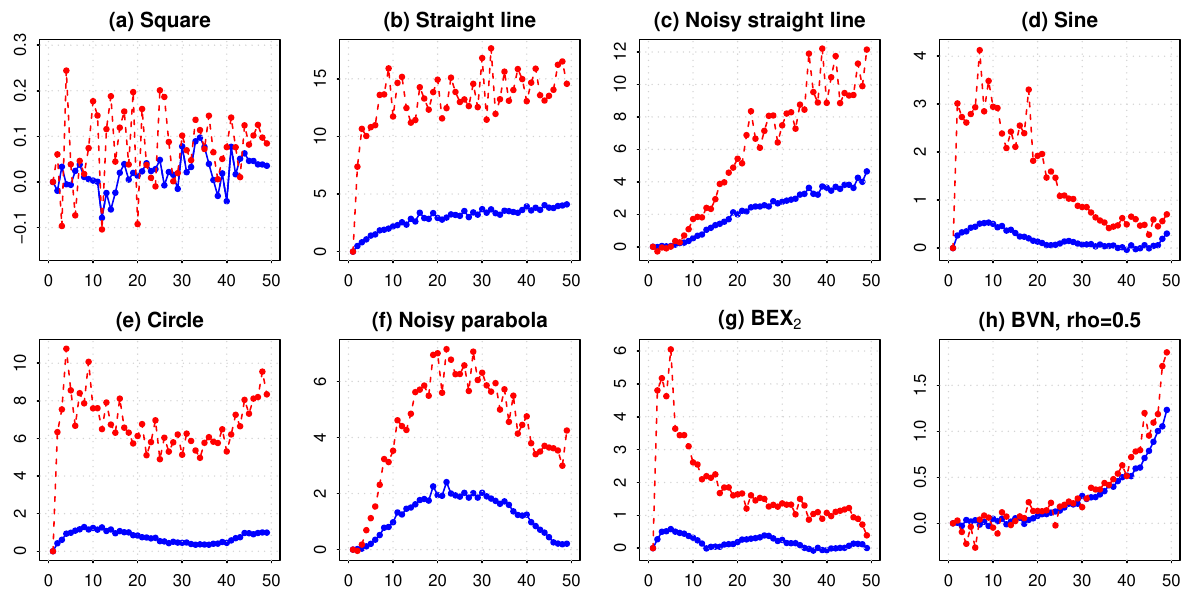}
\caption{Plot of \(\bar{z}_1, \ldots, \bar{z}_{49}\) for different distributions. The blue and red points correspond to \(T^{cor}\) and \(T^{phi}\) statistics, respectively.}
\label{fig:zk_plots2}
\end{figure}

As shown in Figure \ref{fig:zk_plots2}(a), the Z-scores for \(T^{phi}\) exhibit higher variance compared to \(T^{cor}\) under independence.
Figures \ref{fig:zk_plots2}(b) through \ref{fig:zk_plots2}(g) demonstrate that \(T^{phi}\) consistently yields higher Z-scores than \(T^{cor}\) for the `Straight line,' `Noisy straight line,' `Sine,' `Circle,' `Noisy parabola,' and `\(BEX_2\)' distributions.
Additionally, the average Z-score plot for \(T^{phi}\) appears to be more wiggly compared to \(T^{cor}\).
In the `BVN, rho=0.5' distribution, the average Z-scores of \(T^{phi}\) are comparable to those of \(T^{cor}\).

For a fixed \(1 \leq i \neq j \leq n\), calculating \(T_{i,j}\) requires counting the values \(a_{i,j}, b_{i,j}, c_{i,j},\) and \(d_{i,j}\).
A straightforward approach to obtain these counts is to check each observation in \(\mathbf{xy}_{1:n}\) individually.
This process involves \(\mathcal{O}(n)\) operations to compute \(T_{i,j}\).
As established in the previous section, this naive approach requires \(\mathcal{O}(n^3)\) operations to compute the test statistics.
However, by applying the algorithm described in Lemma \ref{lem:complexity}, which computes \(T_{i,j}\) for all \(1 \leq j \neq i \leq n\) simultaneously in \(\mathcal{O}(n \log n)\) operations for a fixed \(i\), the computational complexity can be significantly reduced.
When this algorithm is implemented, calculating \(T_{[1]}, \ldots, T_{[n-1]}\) collectively will require \(\mathcal{O}(n^2 \log n)\) operations.
As a result, the complexity of computing the test statistics \(\Psi_n\) is reduced to \(\mathcal{O}(n^2 \log n)\).

\section{Performance on simulated datasets}

In this section, we compare the performance of our proposed test methods on various simulated datasets with that of existing tests in the literature.
From the proposed general testing framework, we select two well-known test statistics - the absolute value of Pearson's correlation and distance correlation - as \(T\).
We denote the resulting test statistic \(\Psi_n\) in these two cases as \(\Psi_n^{cor}\) and \(\Psi_n^{dcor}\), respectively.
We denote the test statistic \(\Psi_n\) for our proposed special test method as \(\Psi_n^{phi}\).
As pointed out in previous sections, the computational complexities of \(\Psi_n^{phi}, \Psi_n^{cor}\), and \(\Psi_n^{dcor}\) are \(\mathcal{O}(n^2\log n), \mathcal{O}(n^3)\), and \(\mathcal{O}(n^3 \log n)\), respectively.
For our proposed tests, we created an R package, `PAIT,’ which is available at \url{https://github.com/angshumanroycode/PAIT}.

From existing tests of independence, we selected the following popular tests: dCor \citep{szekely2007measuring}, HSIC \citep{gretton2007kernel}, HHG \citep{heller2013consistent}, MIC \citep{reshef2011detecting}, Chatterjee's correlation (xicor) \citep{chatterjee2021new} and BET \citep{zhang2019bet}.
For performing dCor, HSIC, HHG, and xicor tests, we used R packages \textit{energy} \citep{energyR}, \textit{dHSIC} \citep{dHSICR}, \textit{HHG} \citep{HHGR}, and \textit{XICOR} \citep{XICORR}, respectively.
We calculated the MIC statistic by using the `mine\_stat' function from the R package \textit{minerva} \citep{minervaR} and subsequently performed a permutation test.
For executing BET test, we used the R code provided in the supplemental material of \citep{zhang2019bet}.
We fixed the nominal level at \(5\%\).
Except for BET, all other tests were performed using the permutation principle, with the number of permutations set to 1000.
The empirical power of a test is determined by calculating the proportion of times it rejects the null hypothesis over 1000 independently generated samples of a specific size from a given distribution.

To ensure our proposed tests adhere to the specified significance level, we calculated the empirical power for two bivariate distributions with independent marginals. 
In the first case, both marginals followed standard normal distributions, while in the second case, they followed uniform distributions over \([0,1]\). 
The significance level was set at 5\%.
Table \ref{tab:sig_level} presents a summary of the empirical power of various tests for sample sizes of 10, 20, and 50.
\begin{table}
    \centering
    \scalebox{0.85}{
    \begin{tabular}{cccccccccc}
        \hline
        \hline 
        \(n\) & \(\Psi_n^{phi}\) & \(\Psi_n^{cor}\) & \(\Psi_n^{dcor}\) & HHG & HSIC & dCor & MIC & XICOR & BET\\
        \hline
        \hline 
        \multicolumn{10}{c}{Independent \(N(0,1)\) marginals}\\
        \hline
         10 & 0.048 & 0.058 & 0.059 & 0.041 & 0.051 & 0.053 & 0.008 & 0.037 & 0.003 \\
         20 & 0.047 & 0.046 & 0.050 & 0.042 & 0.050 & 0.034 & 0.046 & 0.040 & 0.034 \\
         50 & 0.044 & 0.050 & 0.041 & 0.055 & 0.055 & 0.053 & 0.060 & 0.061 & 0.053 \\
         \hline
        \multicolumn{10}{c}{Independent \(U(0,1)\) marginals}\\
        \hline
         10 & 0.051 & 0.047 & 0.047 & 0.052 & 0.048 & 0.045 & 0.009 & 0.032 & 0.000 \\
         20 & 0.050 & 0.049 & 0.050 & 0.041 & 0.040 & 0.054 & 0.044 & 0.037 & 0.040 \\
         50 & 0.052 & 0.053 & 0.053 & 0.049 & 0.048 & 0.057 & 0.055 & 0.050 & 0.038 \\
         \hline 
         \hline 
    \end{tabular}
    }
    \caption{Empirical power of test methods under independent marginals.}
    \label{tab:sig_level}
\end{table}
All test methods, except MIC, XICOR, and BET, maintained an empirical power close to 5\% across all sample sizes.
MIC, XICOR, and BET exhibited conservative behavior for small sample sizes but approached the nominal significance level as the sample size increased.

First, we considered eight jointly continuous distributions, namely `Doppler', `\(BEX_2\)', `Lissajous curve A', `Lissajous curve B', `Rose curve', `Spiral', `Tilted square', and `Five clouds'.
Their descriptions are provided in Table \ref{tab:8_distributions}.
Scatter plots for each of these distributions are presented in Figure \ref{fig:8_distributions_datasets} in Appendix B.
\begin{table}[h]
\centering
\scalebox{0.75}{
\begin{tabular}{l l}
\hline
\hline
Distribution & Description \\
\hline
\hline
(a) Doppler & \(X\sim U(0,1/2), Y=\sqrt{X}\sin(1/X)\)\\
\hline
(b) \(BEX_2\) & As described in Section 2\\
\hline
(c) Lissajous curve A & \(\Theta\sim U(0, 2\pi), X=\sin(3\Theta+3\pi/4)+\epsilon_X, Y=\sin(\Theta)+\epsilon_Y\)\\
\hline
(d) Lissajous curve B & \(\Theta\sim U(0, 2\pi), X=\sin(4\Theta+\pi/2)+\epsilon_X, Y=\sin(3\Theta)+\epsilon_Y\)\\
\hline
(e) Rose curve & \(\Theta\sim U(0, 2\pi), R=\cos(4\Theta), X=R\cos(\Theta)+\epsilon_X/2, Y=R\sin(\Theta)+\epsilon_Y/2\)\\
\hline
(f) Spiral & \(\Theta\sim U(0, 4\pi), R=\Theta/(2\pi), X=R\cos(\Theta)+\epsilon_X, Y=R\sin(\Theta)+\epsilon_Y\)\\
\hline
(g) Tilted square & \(U,V\stackrel{\text{i.i.d.}}{\sim} U(-1, 1), X=\cos(\pi/3)U-\sin(\pi/3)V, Y=\sin(\pi/3)U+cos(\pi/3)\)V\\
\hline
(h) Five clouds & \((U,V)\) uniform on \(\{(0,0),(0,1),(0.5,0.5),(1,0),(1,1)\}\),\\
& \(X=U+2\epsilon_X, Y=V+2\epsilon_Y\)\\
\hline
\hline
\end{tabular}
}
\caption{The description of `Doppler', `\(BEX_2\)', `Lissajous curve A', `Lissajous curve B', `Rose curve', `Spiral', `Tilted square', and `Five clouds' distributions.  Here \(\epsilon_X, \epsilon_Y\) are i.i.d. random variables following \(N(0,1/30^2)\), independent of \(U, V\), and \(\Theta\).}
\label{tab:8_distributions}
\end{table}
It can be observed from the table that in the `Doppler' distribution, \(Y\) is directly a function of \(X\).
In `Lissajous curve A', `Lissajous curve B', `Rose curve', and `Spiral', \(X\) and \(Y\) are related to each other in form of an implicit function.
`Tilted square' and `Five clouds' are two distributions where \(X\) and \(Y\) are dependent, but their correlation is 0.

The empirical powers of all test methods on these eight distributions are presented in Figure \ref{fig:8_distributions_power}.
\begin{figure}[h]
\centering
\includegraphics[width=\linewidth]{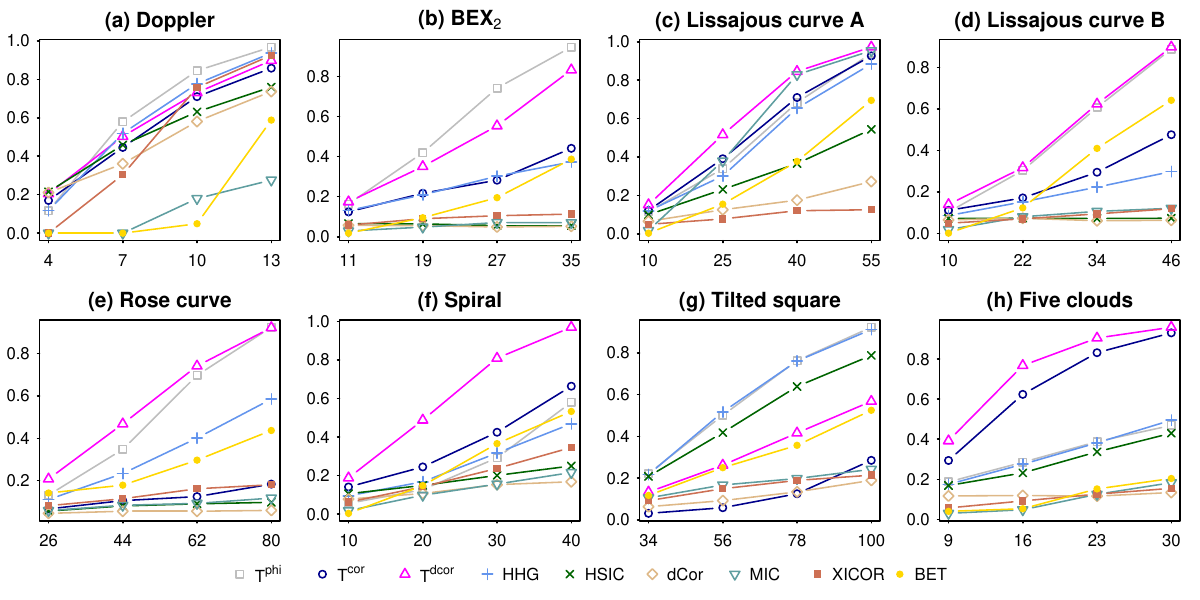}
\caption{The empirical powers of different test methods on `Doppler', `\(BEX_2\)', `Lissajous curve A', `Lissajous curve B', `Rose curve', `Spiral', `Tilted square', and `Five clouds' distributions.}
\label{fig:8_distributions_power}
\end{figure}
It can be seen from this figure that \(\Psi_n^{phi}\) performed best in both `Doppler' and `\(BEX_2\)' and performed well in all other distributions.
\(\Psi_n^{cor}\) performed fairly well in all distributions except for the `Rose curve' and `Tilted square.'
\(\Psi_n^{dcor}\) performed well in all cases, and in `Lissajous curve A', `Lissajous curve B', `Rose curve', `Spiral', `Five clouds', it performed the best.
HHG yielded satisfactory power overall except for `Lissajous curve B'.
The power of HSIC is satisfactory only in the `Doppler' and `Tilted square' distributions, not otherwise.
dCor, MIC, xicor did not perform well.
BET performed well for larger sample sizes in `Lissajous curve B', `Rose curve' and in `Spiral'.

Next, we considered four datasets to assess the performance of our methods under different levels of noise.
In these examples, \(X=f(Z)+\lambda\epsilon_X\) and \(Y=g(Z)+\lambda\epsilon_Y\), where \(f\) and \(g\) are continuous non-constant functions, \(Z\) is a random variable, \(\lambda\) is a non-negative constant, and \(\epsilon_X\) and \(\epsilon_Y\) are i.i.d. normal random variables that are also independent of \(Z\).
As \(\lambda\) increases, the noise level increases.
In particular, we considered 4 distribution named - `Parabola \(\lambda\)', `Circle \(\lambda\)', `Sine \(\lambda\)' and `Laminscate \(\lambda\)', descriptions of each of these can be found in Table \ref{tab:noise_distribution}.
\begin{table}[h]
\centering
\scalebox{0.85}{
\begin{tabular}{l l}
\hline
\hline
Distribution & Description \\
\hline
\hline
(a) Parabola \(\lambda\) & \(U\sim U(-1,1), X=U^2+2\lambda\epsilon_X, Y=U+2\lambda\epsilon_Y\)\\
\hline
(b) Circle \(\lambda\) & \(\Theta\sim U(0,2\pi), X=\cos(\Theta)+6\lambda\epsilon_X, Y=\sin(\Theta)+6\lambda\epsilon_Y\)\\
\hline
(c) Sine \(\lambda\) & \(\Theta\sim U(0,2\pi), X=\Theta+3\lambda\epsilon_X, Y=\sin(\Theta)+3\lambda\epsilon_Y\)\\
\hline
(d) Laminscate \(\lambda\) & \(\Theta\sim U(0,2\pi), X=\frac{\cos(\Theta)}{(1+\sin(\Theta))^2}+2\lambda\epsilon_X, Y=\frac{\sin(\Theta)\cos(\Theta)}{(1+\sin(\Theta))^2}+2\lambda\epsilon_Y\)\\
\hline
\hline
\end{tabular}
}
\caption{The descriptions of `Parabola \(\lambda\)', `Circle \(\lambda\)', `Sine \(\lambda\)' and `Laminscate \(\lambda\)' distributions. Possible values of lambda are \(\lambda=0,1,2\). Here, \(\epsilon_X, \epsilon_Y\) are i.i.d. random variables following \(N(0,1/60^2)\), independent of \(U\) and \(\Theta\).}
\label{tab:noise_distribution}
\end{table}
We considered three noise levels by taking \(\lambda=0,1,2\).
Scatter plots of the twelve (\(4\times 3=12\)) distributions are presented in Figure \ref{fig:noise_distribution_plot} in the Appendix B.

The empirical powers of each of these tests are presented in Figure \ref{fig:noise_distribution_power}.
\begin{figure}[h]
\centering
\includegraphics[width=\linewidth]{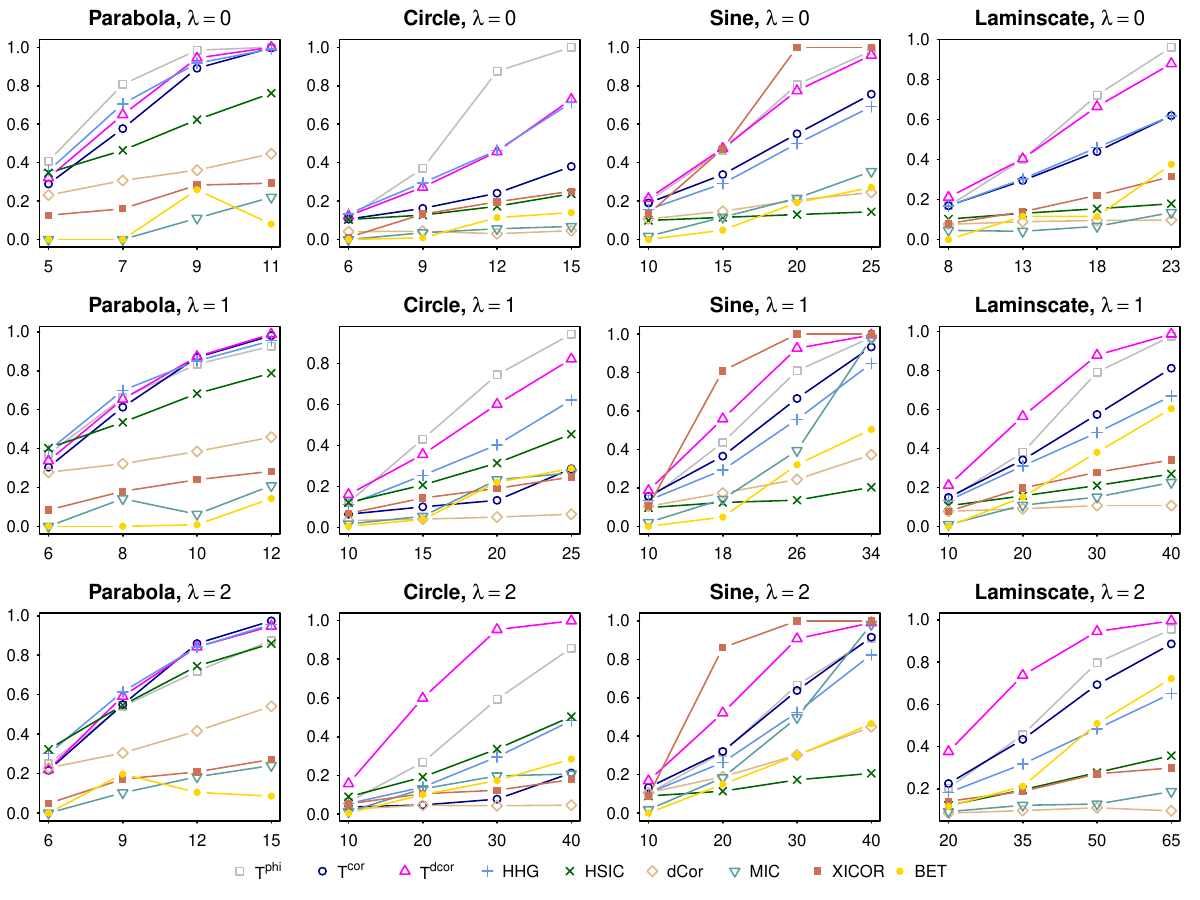}
\caption{The empirical power of different test methods on `Parabola \(\lambda\)', `Circle \(\lambda\)', `Sine \(\lambda\)' and `Laminscate \(\lambda\)' distributions for \(\lambda=0,1,2\).}
\label{fig:noise_distribution_power}
\end{figure}
It can be observed from this figure that \(\Psi_n^{phi}\) performs best or next to best in the absence of noise.
Even in the presence of noise, the power of \(\Psi_n^{phi}\) is quite good.
As the noise level increases, \(\Psi_n^{dcor}\) takes the lead.
The overall performance of \(\Psi_n^{cor}\) is competitive except for the `Circle \(\lambda\)' case.
HHG performed well but it lagged behind our proposed tests most of the time.
HSIC performed well in `Parabola \(\lambda\)' and `Circle \(\lambda\)' distributions, but didn't achieve satisfactory power in other distributions.
dCor had somewhat competitive power in `Parabola \(\lambda\)', but not in other distributions.
MIC had low power overall, except for higher sample sizes in the `Sine \(\lambda\)' distributions.
Xicor also performed poorly except for `Sine \(\lambda\)', where it performed the best.
This is not unexpected from xicor, as it is mainly designed to detect whether \(Y\) is a function of \(X\).
BET performed somewhat well in `Laminscate \(\lambda\)' distribution.

\section{Performance on real dataset}

To demonstrate the application of our proposed multiscale framework, we considered an interesting dataset known as Galton's peas dataset.
This dataset is available under the name \textit{peas} in the R package \textit{psychTools}.
Collected in 1875, this dataset comprises 700 observations of the average diameters of sweet peas from both mother plants and their offspring.
An intriguing aspect of this dataset, pointed out by \cite{chatterjee2021new}, is that the mean diameter of peas produced by a mother plant can be viewed as a function of the mean diameter of their offspring.
However, the opposite does not hold.
When tested on the full dataset, our proposed tests, as well as other competing methods, rejected the null hypothesis of independence.
To compare the performance of these test methods, we evaluated their empirical power using random subsets of the dataset.
Four sample sizes were considered: 10, 20, 30, and 40.
For each sample size, we drew random samples without replacement (SRSWOR) from the full dataset and applied all test methods.
This step was repeated independently 1,000 times to compute the empirical power of each method. 
As in the previous section, we fixed the number of permutations at 1,000 and set the significance level at 5\%.  
Table \ref{tab:galton} summarizes the empirical power of all tests across the different sample sizes.
\begin{table}[h]
    \centering
    \scalebox{0.85}{
    \begin{tabular}{cccccccccc}
        \hline
        \hline 
        \(n\) & \(\Psi_n^{phi}\) & \(\Psi_n^{cor}\) & \(\Psi_n^{dcor}\) & HHG & HSIC & dCor & MIC & XICOR & BET\\
        \hline
        \hline 
         10 & 0.151 & 0.238 & 0.218 & 0.191 & 0.113 & 0.165 & 0.041 & 0.128 & 0.313 \\
         20 & 0.272 & 0.554 & 0.565 & 0.360 & 0.175 & 0.271 & 0.197 & 0.197 & 0.614 \\
         30 & 0.463 & 0.858 & 0.870 & 0.604 & 0.233 & 0.426 & 0.373 & 0.231 & 0.806 \\
         40 & 0.654 & 0.980 & 0.981 & 0.783 & 0.314 & 0.556 & 0.714 & 0.286 & 0.921 \\
         \hline 
         \hline 
    \end{tabular}
    }
    \caption{The empirical power of different test methods on Galton's pea dataset.}
    \label{tab:galton}
\end{table}
The results indicate that \(\Psi_n^{cor}\) and \(\Psi_n^{dcor}\) outperformed all other methods for sample sizes of 30 and 40.
BET also performed very well.
HHG ranked next in performance.
Despite the presence of a huge number of ties in the dataset, \(\Psi_n^{phi}\) demonstrated good performance.
Other methods lagged behind the ones mentioned above.

\section{A visual tool}

Our proposed method works by collecting information on the magnitude of dependence present within the neighborhoods of different sizes around the observation points.
The Z-scores \(z_1,\ldots,z_{n-1}\) calculated based on the sample observations provide insights into the strength of dependence across these different neighborhood sizes.
This is already visually illustrated in Figures \ref{fig:zk_plots} and \ref{fig:zk_plots2} for the distributions described in Table \ref{tab:zk_plots_datasets} and in Figure \ref{fig:zk_plots_datasets}.
Similarly, we can exploit this visual tool to real-world datasets for visualizing the strength of dependence across these different neighborhood sizes.

We illustrate this using the yeast gene expression dataset, initially examined by \cite{spellman1998comprehensive}.
The original study analyzed transcript levels of 6,223 yeast genes across 23 consecutive time points.
For our analysis, we used the R package \textit{minerva}, which provides a revised version of the dataset with 4,381 genes, excluding those with missing observations.
For this illustration, we selected four well-behaved genes—`YBL003C', `YDR224C', `YGR004C', and `YHR218W'—which were also chosen by \cite{chatterjee2021new} to demonstrate XICOR's effectiveness in detecting functional associations. 
Additionally, we selected four genes—`YBR198C', `YBR300C', `YEL049W', and `YEL057C'—which rank among those with the lowest XICOR values.
For these eight genes, we calculated the Z-scores \(z_1,\ldots,z_{22}\) based on \(T^{phi}\) statistics.
The expression of these genes and corresponding Z-scores are plotted in Figure \ref{fig:expression_zscores_plot}.
We applied a 4-point moving average to smooth the Z-score plot for better visualization.
\begin{figure}[h]
\centering
\includegraphics[width=\linewidth]{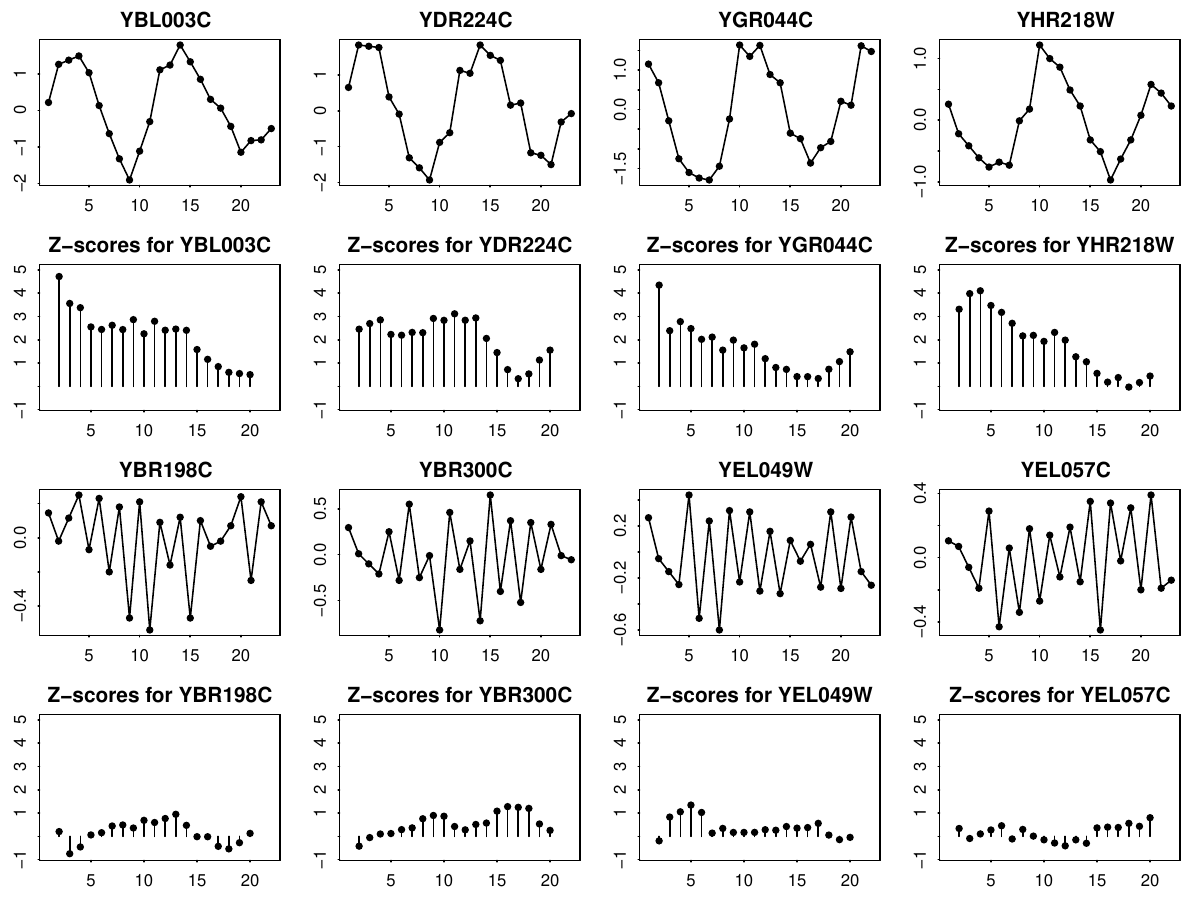}
\caption{Gene expressions for a few selected genes from yeast gene expression data \citep{spellman1998comprehensive}, along with their moving average Z-score plots.}
\label{fig:expression_zscores_plot}
\end{figure}

For the four well-behaved genes, it is evident that smaller neighborhoods capture a significant amount of dependence information.
However, as the neighborhood size increases, the amount of information gradually decreases.
In contrast, for the other four genes, dependence information remains absent across all neighborhood sizes.

\section{Discussion and conclusion}

We proposed a general framework for testing independence, which can utilize existing independence tests in different neighborhoods of the dataset and combine the results in a meaningful way.
This approach has been shown to enhance performance when the variables are explicitly or implicitly functionally dependent.
Additionally, we introduced a novel test of independence that leverages a similar framework.
It is important to note that multiple approaches can be taken towards selecting neighborhoods, and both the power and the complexity of the resulting method can vary based on it.
Therefore, an optimal selection of neighborhoods could be a potential direction for future research.
Our proposed method assigns equal weight to each neighborhood, from nearest to farthest.
The potential benefits of using varying weights for different neighborhoods could be explored in future research.

\appendix

\section*{Appendix A}

\begin{lemma}
    \label{lem:complexity}
    For the proposed special test method, for a fixed \(1 \leq i \leq n\), the computation of \(T_{i,j}\) for all \(1 \leq j \neq i \leq n\) can be done collectively in \(\mathcal{O}(n \log n)\) operations.
\end{lemma}

\begin{proof}
    Let us fix \(i\) such that \(1\leq i\leq n\).
    At first, we will prove that the sequence \(a_{i,1}, \dots, a_{i,i-1}, a_{i,i+1}, \dots, a_{i,n}\) can be computed in \(\mathcal{O}(n \log n)\) operations.

    For \(1\leq j\neq i\leq n\), \(d^x_j=|x_j-x_i|\) and \(d^y_j=|y_j-y_i|\) can be computed in \(\mathcal{O}(n)\) operations.
    As \((x_1,y_1),\ldots,(x_n,y_n)\) are \(n\) i.i.d. observations from a continuous distribution, there are no ties in \(x\)-coordinate or in \(y\)-coordinate, and hence \(d^x_j\) and \(d^y_j\) are positive for all \(1\leq j\neq i\leq n\) without any ties.
    
    Using the merge sort algorithm, a permutation \(\omega\) of the sequence \((1, \ldots, i-1, i+1, \ldots, n)\) can be computed in \(\mathcal{O}(n \log n)\) operations, such that \(d^x_{\omega(1)} \leq \cdots \leq d^x_{\omega(i-1)} \leq d^x_{\omega(i+1)} \leq \cdots \leq d^x_{\omega(n)}\).

    Next, by iterating through \(j \in \{1, \ldots, i-1, i+1, \ldots, n\}\) in ascending order, we can store in an array \(q\) those \(j\) values that satisfy \(x_{\omega(j)} < x_i\) and \(y_{\omega(j)} > y_i\).
    As a result, for all \(k, l \in \{1, \ldots, \text{length}(q)\}\), we have \(q[k] < q[l]\) whenever \(k < l\), with \(x_{\omega(q[k])} < x_i\) and \(y_{\omega(q[k])} > y_i\).
    The remaining integers in \(\{1, \ldots, i-1, i+1, \ldots, n\}\) that do not satisfy the conditions for \(q\) can be stored in another array \(r\) in ascending order.
    Thus, for all \(k, l \in \{1, \ldots, \text{length}(r)\}\), we have \(r[k] < r[l]\) whenever \(k < l\), with \(x_{\omega(r[k])} > x_i\) or \(y_{\omega(r[k])} < y_i\).
    This iterative checking and storing process requires only \(\mathcal{O}(n)\) operations.
    Note that \(q\) and \(r\) have no common elements, and together they contain all elements from the set \(\{1, \ldots, i-1, i+1, \ldots, n\}\).

    Given an array of real numbers \([s_1, \ldots, s_m]\), the `surpasser count' is defined as an array of integers \([t_1, \ldots, t_m]\) of the same length, where \(t_j = \sum_{k=j+1}^m I[s_j < s_k]\) for \(1 \leq j \leq m\).
    Using the merge sort technique, this can be computed in \(\mathcal{O}(m \log m)\) operations \citep[see][Chapter 2]{bird2010pearls}.
    Instead of using the surpasser count algorithm directly, we'll use the `trail count'.
    For a given array of real numbers \([s_1, \ldots, s_m]\), the trail count is defined as an array of integers \([t_1, \ldots, t_m]\) of the same length, where \(t_j = \sum_{k=1}^{j} I[s_j \geq s_k] = 1 + \sum_{k=1}^{j-1} I[s_j > s_k]\) for \(1 \leq j \leq m\).
    Essentially, the trail count is the surpasser count applied to the reversed array, thus maintaining the same complexity.

    We apply the trail count to the array \([d^y_{\omega(1)}, \ldots, d^y_{\omega(i-1)}, d^y_{\omega(i+1)}, \ldots, d^y_{\omega(n)}]\), resulting in the array, say, \([w_1, \ldots, w_{i-1}, w_{i+1}, \ldots, w_n]\).
    It's easy to verify that \(w_k\) represents the number of observations, excluding \((x_i, y_i)\), that are within the neighborhood \(N_{i,\omega(k)}\) for \(1 \leq k \neq i \leq n\).
    Next, we apply the trail count to the array \([d^y_{\omega(q[1])}, \ldots, d^y_{\omega(q[\text{length}(q)])}]\), resulting in the array, say, \([u_1, \ldots, u_{\text{length}(q)}]\).
    It can be verified that \(u_k\) is the number of observations that fall within the set \(\{(x,y) : (x,y) \in N_{i,\omega(q[k])}, x < x_i, y > y_i\}\) for \(1 \leq k \leq \text{length}(q)\).
    Thus, \(a_{i,q[k]} = u_k\) for \(1 \leq k \leq \text{length}(q)\).
    Similarly, applying the trail count to the array \([d^y_{\omega(r[1])}, \ldots, d^y_{\omega(r[\text{length}(r)])}]\) results in the array, say, \([v_1, \ldots, v_{\text{length}(r)}]\).
    It can be verified that \(v_k\) represents the number of points in the set \(\{(x,y) : (x,y) \in N_{i,\omega(r[k])}, x > x_i \text{ or } y < y_i\}\) for \(1 \leq k \leq \text{length}(r)\).
    Therefore, \(a_{i,r[k]} = w_{r[k]} - v_k\) for \(1 \leq k \leq \text{length}(r)\).
    Hence, all \(a_{i,j}\) for \(1 \leq j \neq i \leq n\) can be computed together in \(\mathcal{O}(n \log n)\) operations.

    Similar steps can be taken to show that for a fixed \(i\), \(b_{i,j},c_{i,j},d_{i,j}\) for \(1\leq j\neq i\leq n\) can be computed collectively in \(\mathcal{O}(n \log n)\) operations.
    Since \(T_{i,j}\) is a function of \(a_{i,j}, b_{i,j}, c_{i,j}\), and \(d_{i,j}\), it directly follows that computing \(T_{i,j}\) collectively for \(1\leq j\neq i\leq n\) also requires \(\mathcal{O}(n \log n)\) operations.
\end{proof}

\section*{Appendix B}

\begin{figure}[H]
\centering
\includegraphics[width=0.85\linewidth]{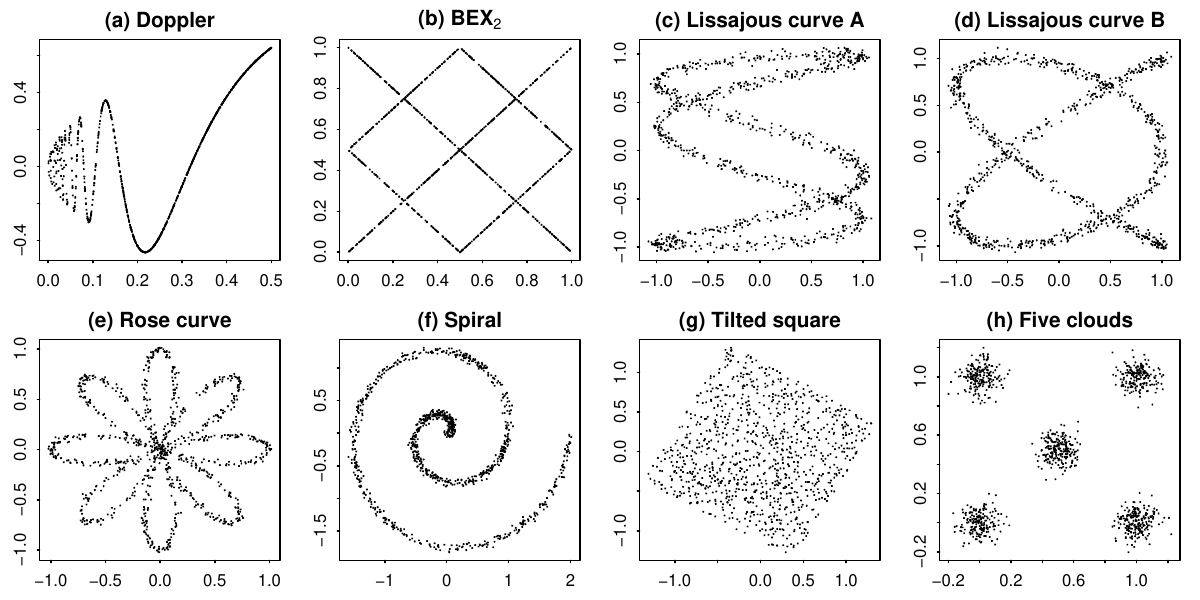}
\caption{The scatter plots of `Doppler', `\(BEX_2\)', `Lissajous curve A', `Lissajous curve B', `Rose curve', `Spiral', `Tilted square' and `Five clouds' distributions.}
\label{fig:8_distributions_datasets}
\end{figure}

\begin{figure}[H]
\centering
\includegraphics[width=0.85\linewidth]{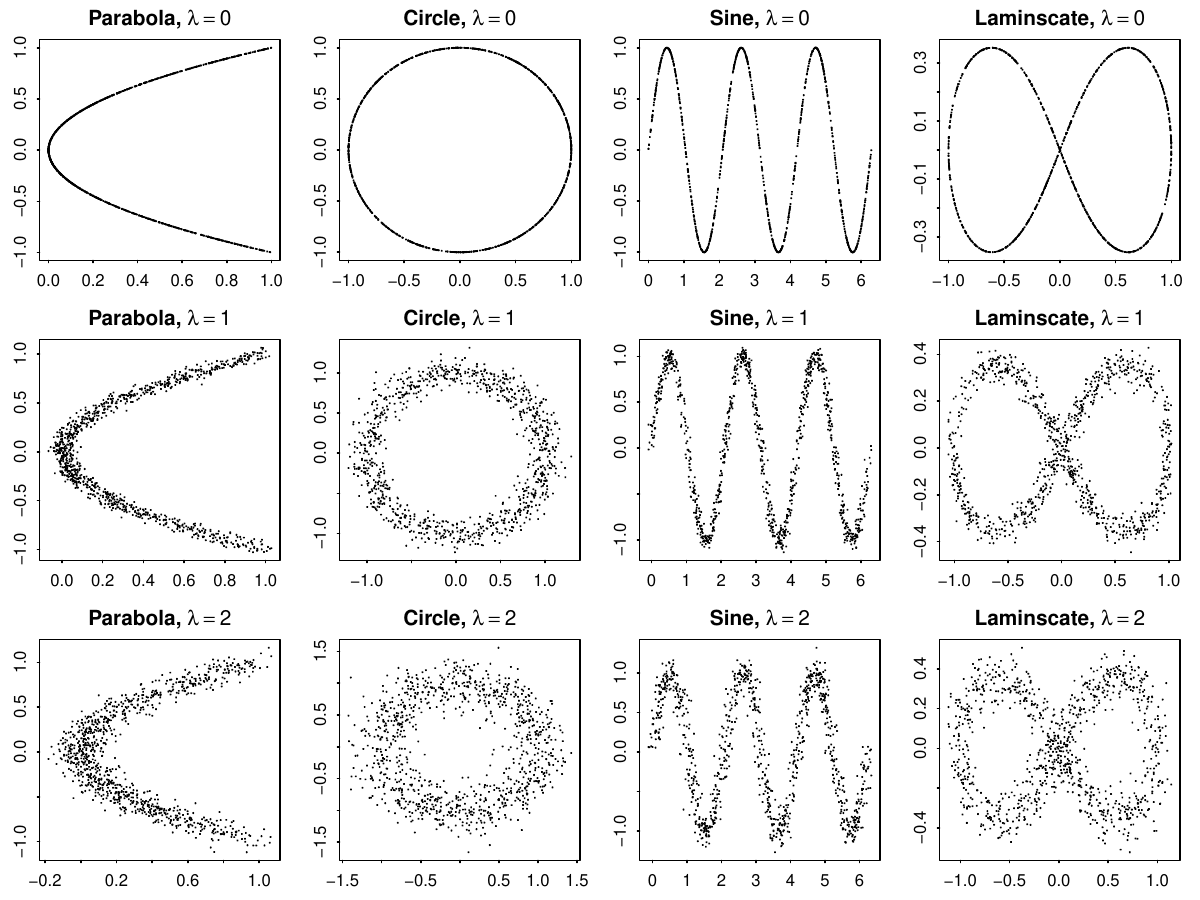}
\caption{The plot of `Parabola \(\lambda\)', `Circle \(\lambda\)', `Sine \(\lambda\)' and `Laminscate \(\lambda\)' distributions for \(\lambda=0,1,2\).}
\label{fig:noise_distribution_plot}
\end{figure}


\bibliographystyle{plainnat}
\bibliography{reference}
\end{document}